\documentclass[aps, prd, reprint, amsmath, amssymb, superscriptaddress,nofootinbib]{revtex4-2}

\usepackage{amsthm}
\newtheorem{theorem}{Theorem}

\usepackage[utf8]{inputenc}
\usepackage[T1]{fontenc}
\usepackage{newtxtext, newtxmath} 
\usepackage{microtype}

\usepackage{graphicx}
\usepackage{bm}          
\usepackage{booktabs}    
\usepackage[x11names]{xcolor}

\usepackage[colorlinks=true, citecolor=RoyalBlue3, linkcolor=RoyalBlue3, urlcolor=RoyalBlue3]{hyperref}

\begin{document}

\title{Positivity of the effective range for finite range attractive potentials with a repulsive core}

\author{Davide Germani}
\email{davide.germani@uniroma1.it} 
\affiliation{Sapienza University of Rome and INFN, Piazzale Aldo Moro 2, I-00185, Italy}
\affiliation{INFN Sezione di Roma, Piazzale Aldo Moro 2, I-00185 Rome, Italy}

\date{\today}

\begin{abstract}
In the phenomenological study of exotic hadrons, the sign of the effective range, $r_0$, is invoked as a criterion to distinguish between compact multiquark configurations (associated with $r_0 < 0$) and loosely bound hadronic molecules ($r_0 > 0$). Motivated by this, we investigate the fundamental constraints on the sign of the effective range for single-channel local interactions. We rigorously prove that for finite-range potentials, characterized by an inner repulsive core and an outer attractive tail, the effective range remains strictly positive provided that the scattering length is greater than the range of the potential ($a > R$).
\end{abstract}

\maketitle

\section{Introduction}
In quantum mechanics, low-energy scattering provides a powerful and universal framework to investigate the underlying interactions between particles without requiring detailed knowledge of the short-range dynamics. When the de Broglie wavelength of the incident particle is significantly larger than the characteristic range of the interaction potential $R$ (i.e., $kR \ll 1$), the scattering process is completely dominated by the S-wave ($l=0$) contribution. The corresponding scattering amplitude $f_0(k)$ can be expressed in terms of the S-wave phase shift $\delta_0(k)$ as \cite{Landau:1991wop,osti_4661960}
\begin{equation}
    f_0(k)=\frac{1}{k\cot\delta_0(k)-ik}\,.
    \label{eq:f0}
\end{equation}
Provided the potential is exponentially bounded, its details in the low-energy regime are completely encapsulated in two fundamental constants through the Effective Range Expansion (ERE) \cite{Bethe,PhysRev.74.92,Landau:1991wop,osti_4661960}:
\begin{equation}
k \cot \delta_0(k) = -\frac{1}{a} + \frac{1}{2} r_0 k^2 + \mathcal{O}(r^3_0k^4) \, ,
\label{eq:ERE}
\end{equation}
where $a$ is the scattering length, and $r_0$ is the effective range. This expansion is also referred to as the \textit{shape-independent approximation}, as virtually any potential can be suitably tuned in range and depth to reproduce the desired values of $a$ and $r_0$ \cite{Blatt:1952ije}.

The scattering length $a$ dictates the zero-energy cross-section, and its sign is intricately linked to the presence of bound or virtual states near the scattering threshold \cite{Macedo-Lima:2023fzp,osti_4661960,Landau:1991wop}. The effective range $r_0$, on the other hand, provides a measure of the spatial extent over which the potential modifies the free wave function. For standard finite-range attractive potentials, particularly those supporting a shallow bound state, it is a well-established mathematical result that the effective range is strictly positive \cite{Smorod,Landau:1991wop}. This stems from the integral representation of $r_0$, where the true wave function within the interaction region is systematically smaller than its asymptotic free counterpart.

Recently, however, the specific value and, more importantly, the \textit{sign} of the effective range have garnered renewed and intense interest in the field of hadron physics. In the ongoing study of exotic hadrons, the sign of $r_0$ is widely discussed as a physical criterion to distinguish between a compact multiquark nature and a loosely bound hadronic molecule \cite{Esposito:2021vhu,Esposito:2025hlp,Baru:2021ldu,Shen:2024npc,Kinugawa:2023fbf}. A molecular structure (analogous to the deuteron \cite{Weinberg:1965zz}) is associated with a positive effective range. Conversely, a compact state produces a negative effective range.

Motivated by this phenomenological discussion, it becomes important to establish rigorous bounds on what types of fundamental interactions can actually produce a negative $r_0$. If one attempts to model the interaction between constituent hadrons using an effective single-channel local potential, it is crucial to understand whether specific shapes can force an inversion of the sign of $r_0$.

This paper aims to contribute to this discussion by investigating the fundamental constraints on the sign of the effective range. Specifically, we focus on a physically relevant class of local, finite-range potentials characterized by an inner repulsive core with an attractive tail. We rigorously prove that, provided the potential has a scattering length greater than its range of action (i.e., $a > R$), the effective range $r_0$ remains strictly positive. This finding imposes a severe mathematical constraint: adding an arbitrary inner repulsion to a single-channel local potential is fundamentally insufficient to yield a negative effective range. Consequently, this restricts the class of potential models that can be legitimately used to describe compact exotic hadrons without invoking explicit coupled-channel dynamics or energy-dependent interactions \cite{Baru:2021ldu,Esposito:2023mxw}.

This paper is organized as follows. In Section \ref{sec:low_en_par}, we review the definitions and the mathematical derivations of the scattering length and the effective range. In Section \ref{sec:proof}, we state and prove our main theorem concerning the strict positivity of $r_0$ for potentials featuring an inner repulsive core. To illustrate the validity and the physical meaning of this result, we apply it to two specific phenomenological models: a piecewise spherical barrier-well potential and a superposition of repulsive and attractive Yukawa potentials. Finally, our conclusions and the phenomenological implications of our findings are summarized in Section \ref{sec:conc}.

\section{Low energy parameters}
\label{sec:low_en_par}
In this section, we review the definitions and derivations of the scattering length and the effective range.

\subsection{Scattering length}
Let $V(\bm{r})$ be the scattering potential, which we assume to be real, local, finite-ranged, and spherically symmetric, i.e., $V(\bm{r})=V(r)$ and $V(r)=0$ for $r>R$, where $R$ defines the range of the potential. The wave function of the scattering particle satisfies the Schr\"odinger equation
\begin{equation}
    -\frac{\hbar^2}{2m} \nabla^2 \psi(r,\theta)+V(r)\psi(r,\theta)=E\psi(r,\theta),
\end{equation}
where $E=\hbar^2k^2/2m$. Due to spherical symmetry, the eigenfunctions can be labeled by the magnitude of the momentum $k=|\bm{k}|$. Since $[V,L_z]=[V,L^2]=0$, the wave function can be decomposed by separating the angular dependence as
\begin{equation}
    \psi_k(r,\theta)=\Omega_k(r)Y_l^m(\theta,\phi)\,,
\end{equation}
where $\Omega_k(r)$ is the radial wave function and $Y_l^m(\theta,\phi)$ are the spherical harmonics. Introducing the reduced radial wave function $u_k(r)$, defined as $\Omega_k(r)=u_k(r)/r$, we obtain the equation
\begin{equation}
    \left(-\frac{\hbar^2}{2m}\frac{d^2}{dr^2}+\frac{\hbar^2}{2m}\frac{l(l+1)}{r^2}+V(r)\right)u_k(r)=\frac{\hbar^2k^2}{2m}u_k(r)\,.
\end{equation}
Requiring $\Omega_k(0)$ to be finite imposes the boundary condition $u_k(0)=0$.

In this work, we focus on low-energy scattering, corresponding to the regime $kR \ll 1$. In this case, it is well known that the partial waves with $l>0$ are unimportant, and the $l=0$ component is dominant \cite{Landau:1991wop,osti_4661960}. Thus, the equation for $u_k$ becomes
\begin{equation}
    \frac{d^2}{dr^2}u_k(r)+\left(k^2-U(r)\right)u_k(r)=0\,,
\end{equation}
with $U(r)=2mV(r)/\hbar^2$. Outside the interaction region, $U(r)=0$, meaning the wave function satisfies the free equation
\begin{equation}
    \left(\frac{d^2}{dr^2}+k^2\right)u_k(r)=0\,\Rightarrow\, u_k(r)=A\sin(kr+\delta_0(k))\,.
\label{eq:free_sch}
\end{equation}
The appearance of the phase shift $\delta_0(k)$ signals that the particle has interacted with the potential for $r\leq R$, so that its wave function differs from the free one. The phase shift can be determined analytically whenever the internal solution is known, by imposing continuity and differentiability at $r=R$. Imposing the continuity of the logarithmic derivative makes the matching independent of the chosen normalization, yielding
\begin{equation}
    \frac{u^\prime_k(r)}{u_k(r)}\bigg|_{r=R^-}
    =\frac{u^\prime_k(r)}{u_k(r)}\bigg|_{r=R^+}
    =k\cot(kR+\delta_0(k)).
\end{equation}
For $k=0$, the equation for $r>R$ reduces to
\begin{equation}
    u^{\prime\prime}_0(r)=0\,,
\end{equation}
whose solution is, without loss of generality,
\begin{equation}
    u_0(r)=A(r-a)\quad\quad r>R\,,
    \label{eq:sol_k0}
\end{equation}
where $a$ and $A$ are integration constants. The usefulness of this parameterization for the homogeneous solution will become apparent shortly. This solution must match the $k\to0^+$ limit of Eq. \eqref{eq:free_sch}, thus leading to
\begin{equation}
 \lim_{k\to0^+}u_k(r)=u_0(r)
\end{equation}
or
\begin{equation}
     \lim_{k\to0^+}k\cot\delta_0(k)=-\frac{1}{a}.
\end{equation}
We have thus obtained the leading order of the effective range expansion, identifying the parameter $a$ as the \textit{scattering length}. This also allows for a simple geometric interpretation. Returning to Eq. \eqref{eq:sol_k0}, we choose the normalization such that $A=-1/a$ (which consequently fixes the normalization of the interior solution), yielding
\begin{equation}
    u_0(r)=1-\frac{r}{a}.
\end{equation}
This function vanishes at $r=a$. The scattering length can therefore be defined as the intercept (or its extrapolation) of the asymptotic expression of the zero-energy wave function $u_0(r)$ \cite{Blatt:1952ije}.

\begin{figure}[h]
    \centering
    \includegraphics[width=0.9\linewidth]{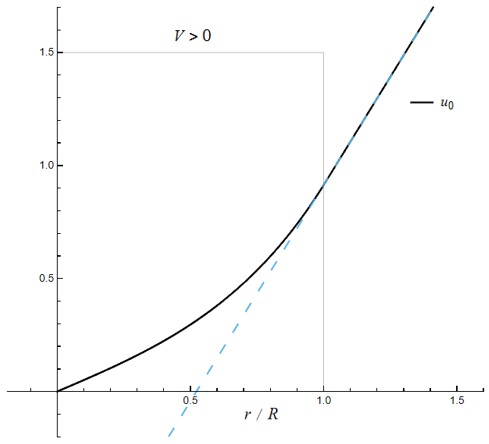}
    \caption{Repulsive potential. The function $u_0$ is shown in black, while the dashed line represents the asymptotic expression. The scattering length is positive, as the line intersects the positive $r-$axis, and lies between $0$ and $R$.}
    \label{fig:rep_pot}
\end{figure}
\begin{figure}[h]
    \centering
    \includegraphics[width=0.9\linewidth]{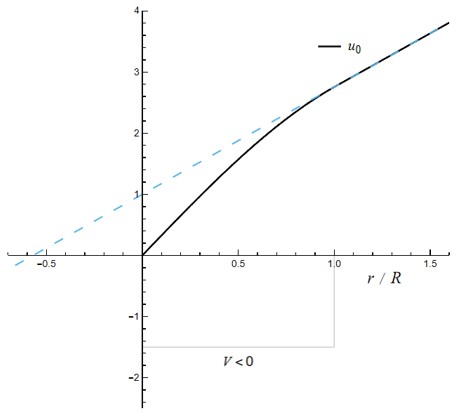}
    \caption{Attractive potential not strong enough to support a bound state. The function $u_0$ is shown in black, while the dashed line represents the asymptotic expression. The scattering length is negative, as the line intersects the negative $r$-axis.}
    \label{fig:att_pot1}
\end{figure}

Repulsive potentials, i.e., $V(r)>0$ for $r<R$, always yield positive scattering lengths. Without loss of generality, we assume $u_0^\prime(0)=1$; \footnote{Since the Schr\"odinger equation is linear, the overall normalization is arbitrary; if $u_k(r)$ is a solution, so is $\alpha u_k(r)$ for any $\alpha \in \mathbb{R}$.} then $u_0(r)$ is convex for $r\in(0,R]$, since
\begin{equation}
    u^{\prime\prime}_0(r)=U(r)\,u_0(r)>0,
\end{equation}
which implies that $u_0(r)$ lies above any of its tangent lines. In particular, it lies above the tangent at $r=R$, which is proportional to $1-r/a$. Hence this line must intersect the horizontal axis at some point in $(0,R)$ (excluding $R$ to avoid singular behavior of $u_0(r)$). 

Attractive potentials, i.e., $V(r)<0$ for $r<R$, on the other hand, may lead to either positive or negative scattering lengths of arbitrarily large magnitude. If the potential is not strong enough to produce a bound state, the zero-energy wave function $u_0(r)$ has no nodes, and its asymptotic linear extrapolation intersects the negative $r$-axis, resulting in $a<0$. If instead a bound state is formed, then $a>0$ and $u_0(r)$ has a node at $r=a$. This reflects the appearance of a discrete bound state with $E<0$, which increases the number of nodes of the continuum wave functions above it. This bound state is accompanied by the appearance of a pole in the scattering amplitude, Eq.~\eqref{eq:f0}, located on the imaginary axis of the complex momentum plane \cite{osti_4661960,Landau:1991wop}.\footnote{We recall that $f_0(k)$ is a meromorphic function in the complex $k$-plane for finite-range potentials \cite{osti_4661960}.} Indeed, retaining only the leading-order term in the expansion of $k\cot\delta_0(k)$, we obtain
\begin{equation}
    f_0(k)=\frac{1}{-\dfrac{1}{a}-ik}\,,
\end{equation}
which exhibits a pole at $\kappa=i/a$, provided that the $\mathcal{O}(k^2)$ terms can be neglected.\footnote{As seen previously, purely repulsive potentials also yield $a>0$. However, in that case $a<R$, meaning that higher-order terms in the ERE \eqref{eq:ERE} cannot be safely neglected.} This condition is met, for example, when $a \gg R$. In this regime, $|\kappa|^{-1} \gg R$, and the next-to-leading order term is suppressed by a factor of $R/a$ (assuming $R$ is the natural length scale of the system). The resulting bound state, with energy
\begin{equation}
    E_B=-B=\frac{\hbar^2\kappa^2}{2m}=-\frac{\hbar^2}{2ma^2}\,,
    \label{eq:B_a}
\end{equation}
where $B$ is the binding energy, is referred to as a \textit{shallow bound state}. Its energy is atypically small compared to the natural energy scales of the system, or equivalently, its spatial extent is much larger than the potential range \cite{Landau:1991wop,osti_4661960}. This is because, outside the interaction region, the bound-state wave function assumes the universal form
\begin{equation}
    u_B(r)\propto e^{-\frac{\sqrt{2mB}}{\hbar}\,r}\,,
\end{equation}
whose spatial extent is much larger than the range of the potential
\begin{equation}
    \frac{\hbar}{\sqrt{2mB}}\gg R,
\end{equation}
which is equivalent to $a\gg R$ using Eq. \eqref{eq:B_a}.
\begin{figure}[b]
    \centering
    \includegraphics[width=0.9\linewidth]{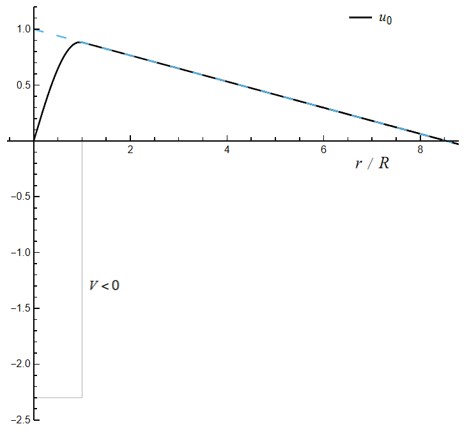}
    \caption{Attractive potential that supports a shallow bound state. The function $u_0$ is shown in black, while the dashed line represents the asymptotic expression. The scattering length is positive and $a \gg R$.}
    \label{fig:att_pot2}
\end{figure}
\begin{figure}[h!]
    \centering
    \includegraphics[width=0.9\linewidth]{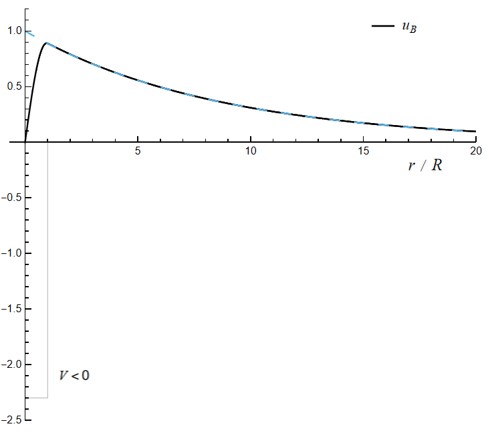}
    \caption{The wave function $u_B$ of the shallow bound state is shown in black. The spatial extent of the wave function is much larger than the range of the potential. The dashed line represents the exponential behavior of the asymptotic expression.}
    \label{fig:b_s}
\end{figure}

\subsection{Effective range}

In this section, we derive the integral representation for the effective range that will be used in the following discussion, following \cite{Bethe,Blatt:1952ije}. Let $u_k(r)$ be the solution of the Schr\"odinger equation
\begin{equation}
    -\frac{d^2}{dr^2}u_k(r)+U(r)\,u_k(r)=k^2u_k(r),
\end{equation}
with $u_k(0)=0$, and let $\chi_k(r)$ be the corresponding asymptotic solution normalized such that $\chi_k(0)=1$
\begin{equation}
    \chi_k(r)=\frac{\sin(kr+\delta_0(k))}{\sin\delta_0(k)}.
\end{equation}
This choice also fixes the normalization of $u_k(r)$, since $u_k(r)=\chi_k(r)$ for $r\geq R$. Clearly $\chi_k(r)$ satisfies the free equation
\begin{equation}
    -\frac{d^2}{dr^2}\chi_k(r)=k^2\,\chi_k(r),
\end{equation}
but, not being physical for $r\leq R$, it is not required to satisfy $\chi_k(0)=0$.

Consider two distinct values $k_1$ and $k_2$ and the respective solutions
\begin{equation}
    -u^{\prime\prime}_{k_1}(r)+U(r)u_{k_1}(r)=k_1^2 u_{k_1}(r),
\end{equation}
\begin{equation}
    -u^{\prime\prime}_{k_2}(r)+U(r)u_{k_2}(r)=k_2^2 u_{k_2}(r).
\end{equation}
Multiplying the first equation by $u_{k_2}$ and the second by $u_{k_1}$, and subtracting, we find
\begin{equation}
    u_{k_1}^{\prime\prime}(r)u_{k_2}(r)-u_{k_2}^{\prime\prime}(r)u_{k_1}(r)
    =(k_2^2-k_1^2)u_{k_1}(r)u_{k_2}(r).
\end{equation}
Integrating from $0$ to $R$ yields
\begin{equation}
    [u^\prime_{k_1}u_{k_2}-u^\prime_{k_2}u_{k_1}]_0^R
    =(k_2^2-k_1^2)\int_0^R u_{k_1}(r)u_{k_2}(r)\,dr.
\label{eq:u_r0}
\end{equation}
Repeating the same steps for $\chi_k(r)$, we obtain
\begin{equation}
    [\chi^\prime_{k_1}\chi_{k_2}-\chi^\prime_{k_2}\chi_{k_1}]_0^R
    =(k_2^2-k_1^2)\int_0^R \chi_{k_1}(r)\chi_{k_2}(r)\,dr.
\label{eq:chi_r0}
\end{equation}
Subtracting Eq. \eqref{eq:u_r0} from Eq. \eqref{eq:chi_r0} and using $u_k^{(\prime)}(R)=\chi_k^{(\prime)}(R)$, we arrive at
\begin{multline}
    k_2\cot\delta_0(k_2)-k_1\cot\delta_0(k_1)
    =\\(k_2^2-k_1^2)\int_0^R \big[\chi_{k_1}(r)\chi_{k_2}(r)-u_{k_1}(r)u_{k_2}(r)\big]\,dr.
\end{multline}
This exact identity is a fundamental equation of scattering theory. Taking the limit $k_1\to0^+$, using
\begin{equation}
   \lim_{k_1\to0^+} k_1\cot\delta_0(k_1)=-\frac{1}{a}\,,
\end{equation}
and relabeling $k_2$ as $k$ for simplicity, we obtain
\begin{equation}
    k\cot\delta_0(k)
    =-\frac{1}{a}
    +k^2\int_0^R \big[\chi_k(r)\chi_0(r)-u_k(r)u_0(r)\big]\,dr.
\end{equation}
Defining
\begin{equation}
    \rho(k)=2\int_0^R \big[\chi_k(r)\chi_0(r)-u_k(r)u_0(r)\big]\,dr,
\end{equation}
we identify the \textit{effective range} as \cite{Bethe,Blatt:1952ije}
\begin{equation}
    r_0=\lim_{k\to0^+}\rho(k)
    =2\int_0^R \big[\chi_0^2(r)-u_0^2(r)\big]\,dr.
\label{eq:r0}
\end{equation}
Since $\chi_0$ and $u_0$ differ only within the interaction region, this integral intuitively quantifies the effective range of action of the potential \cite{osti_4661960,Blatt:1952ije}.

\section{\texorpdfstring{On the sign of $r_0$}{On the sign of r0}}
\label{sec:proof}

It is a well-known result that for purely attractive potentials supporting a shallow bound state, the effective range is positive \cite{Landau:1991wop,Smorod}. Here, we generalize this result to potentials featuring a repulsive core.

\begin{theorem}
Let $V(\bm{r})$ be a real, local, finite-range, spherically symmetric scattering potential, such that $V(\bm{r})=V(r)$ and $V(r)=0$ for $r>R$. Assume there exists $r_1\in(0,R)$ such that $V(r)\geq0$ for $r\in(0,r_1)$, and $V(r)\leq0$ for $r\in(r_1,R)$. Furthermore, assume $a>R$. Then, $r_0>0$.
\end{theorem}
\begin{proof}
Define $\Delta(r)=\chi_0(r)-u_0(r)$. This function is $C^1(\mathbb{R})$, vanishes identically for $r\geq R$, and satisfies $\Delta(0)=1$. Moreover,
\begin{equation}
    \Delta^{\prime\prime}(r)=-u_0^{\prime\prime}(r)=-U(r)\,u_0(r).
\end{equation}
The function $u_0(r)$ has a single node. Since for $r\geq R$ we have $u_0(r)=\chi_0(r)=1-r/a$, and $\chi_0(r)$ has a node at $r=a> R$, it follows that $u_0(r)>0$ for $r\leq R$. Therefore,
\begin{equation}
    \Delta^{\prime\prime}(r)=\begin{cases}
        \leq0 & r\in(0,r_1)\\
        \geq0 & r\in(r_1,R)\,.
    \end{cases}
\end{equation}
Hence $\Delta(r)$ is convex on $(r_1,R)$, and since $\Delta(R)=\Delta^\prime(R)=0$, we conclude that $\Delta(r)\geq0$ in $[r_1,R]$. On $(0,r_1)$, $\Delta(r)$ is concave, with both $\Delta(0)$ and $\Delta(r_1)$ positive; therefore $\Delta(r)\geq0$ also on $[0,r_1]$. Altogether,
\begin{equation}
    \Delta(r)\geq0 
    \,\Rightarrow\,
    \chi_0(r)\geq u_0(r)>0 \quad r\in[0,R].
\end{equation}
Consequently, $\chi_0^2(r)\geq u_0^2(r)$ over $[0,R]$. Since $\chi_0(0)>u_0(0)$ and both are continuous functions, it follows immediately from Eq. \eqref{eq:r0} that $r_0>0$.
\end{proof}

\subsection{Examples}
Below we present two examples of potentials to which the theorem applies. The simplest case is a piecewise continuous potential consisting of a spherically symmetric finite barrier followed by a potential well. This potential can also be solved analytically in a specific limit, as we will show. As a second example, we consider the sum of an attractive and a repulsive Yukawa potential, multiplied by a Heaviside step function to guarantee a finite range.

\subsubsection{Spherical barrier and well}
We consider the potential
\begin{equation}
    V_w(r)=\begin{cases}
        V_1 & r\in[0,r_1),\\
        -V_2 & r\in [r_1,R),\\
        0 & r>R,
    \end{cases}
    \label{eq:vr_well}
\end{equation}
with $V_i>0$. For generality, we rewrite the problem in terms of the dimensionless variable
\begin{equation}
    \bar{r}=\frac{r}{R}\,\,.
\end{equation}
Thus
\begin{equation}
    \bar{U}_w(\bar{r})=\frac{2m}{\hbar^2}R^2V_w(\bar{r})
\end{equation}
and
\begin{equation}
    \bar{u}_k(\bar{r})=\sqrt{R}\,u_k(\bar{r})\,\,,
\end{equation}
so that the equation to solve for $a=R\,\bar{a}$ and $r_0=R\,\bar{r}_0$ becomes\footnote{This is equivalent to setting $\hbar=1$ and $m=1/2$.}
\begin{equation}
    -\bar{u}_0^{\prime\prime}(\bar{r})+\bar{U}(\bar{r})\bar{u}_0(\bar{r})=0\,.
\label{eq:well}
\end{equation}
To solve this differential equation numerically, we implemented the method proposed in \cite{Macedo-Lima:2023fzp} using \texttt{Wolfram Mathematica}. Using, for example, the parameters in Tab. \ref{tab:par_well}, we obtain $\bar{a}=1.17$ and $\bar{r}_0=0.33$. It is instructive to compare this result with the case of a purely attractive spherical well $V^*$ of the same range $R$. By selecting the value reported in Tab. \ref{tab:par_well}, we find $\bar{a}^*=\bar{a}$ but $\bar{r}^*_0=0.32$. The inclusion of a repulsive core leads to an effective range larger than that of a purely attractive potential with the same range and scattering length. This is consistent with the behavior shown in Fig. \ref{fig:u0_well}, where the wave function associated with $V_w$ consistently lies below the corresponding function for $V^*_w$. This suppression is due to the classically forbidden region $r<r_1$.

\begin{table}[h]
    \centering
    \begin{tabular}{|cccc|ccc|}
    \hline
        $\bar{r}_1$ & $\bar{U}_1$ & $\bar{U}_2$ & $\bar{U}^*$ & $\bar{a}$ & $\bar{r}_0$ & $\bar{r}_0^*$ \\
        \hline\hline
         0.2 & 15 & 8 & 7.3667 & 1.17 & 0.33 & 0.32 \\
    \hline
    \end{tabular}
    \caption{Parameters used for the spherical barrier and well problem, and the corresponding scattering length and effective range.}
    \label{tab:par_well}
\end{table}

\begin{figure}[h]
    \centering
    \includegraphics[width=0.9\linewidth]{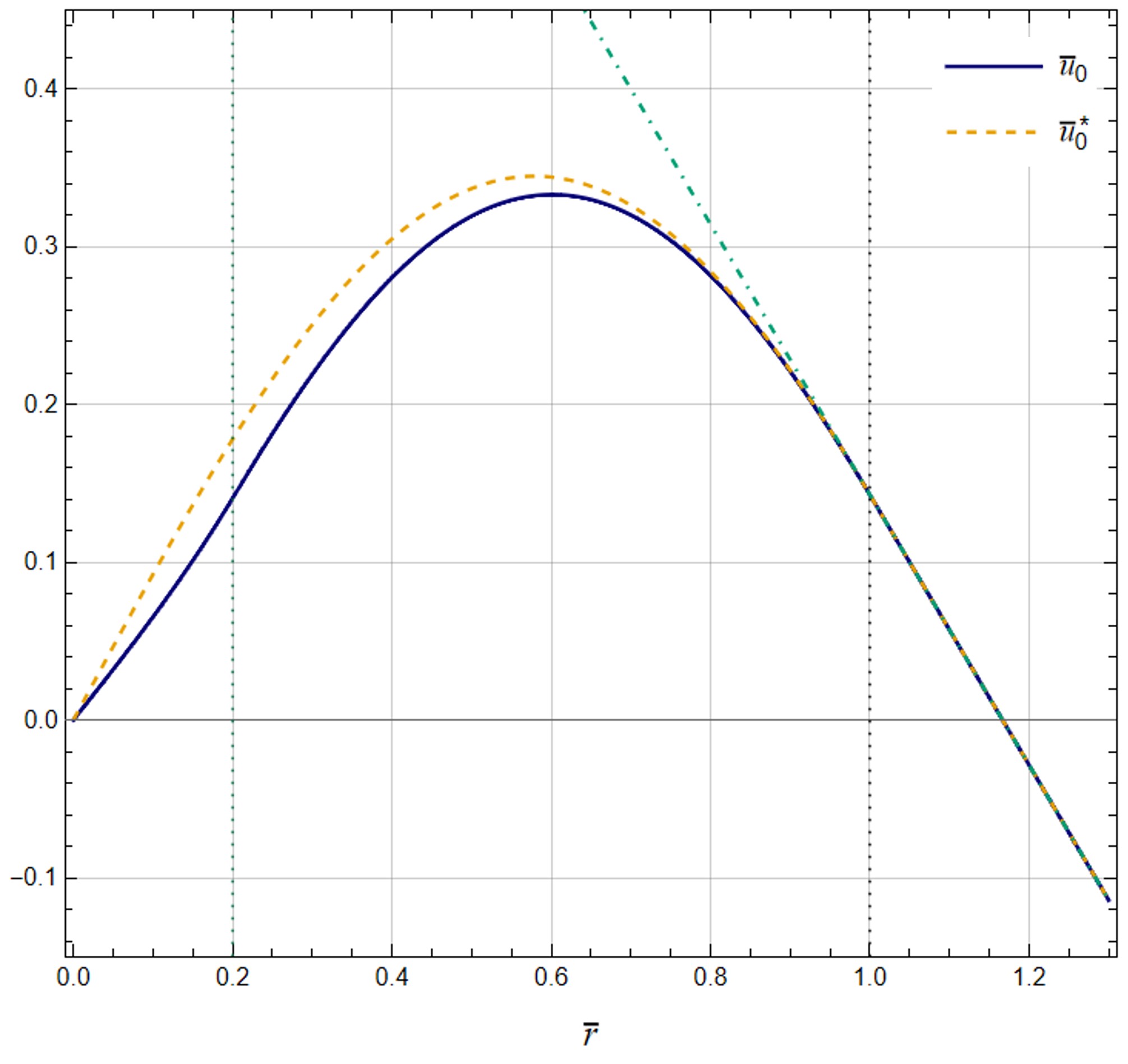}
    \caption{Solid line: zero-energy wavefunction for the potential in Eq.~\eqref{eq:vr_well} with the parameters listed in Tab.~\ref{tab:par_well}. The dashed line indicates the solution for the purely attractive problem. The dash-dotted line represents the asymptotic solution $1-\bar{r}/\bar{a}$, which is common to both solutions since they share the same scattering length. The vertical dotted lines indicate the end of the repulsive region ($\bar{r}=\bar{r}_1$) and the range of the potential ($\bar{r}=1$), respectively.}
    \label{fig:u0_well}
\end{figure}
\begin{figure}[h]
    \centering
    \includegraphics[width=0.9\linewidth]{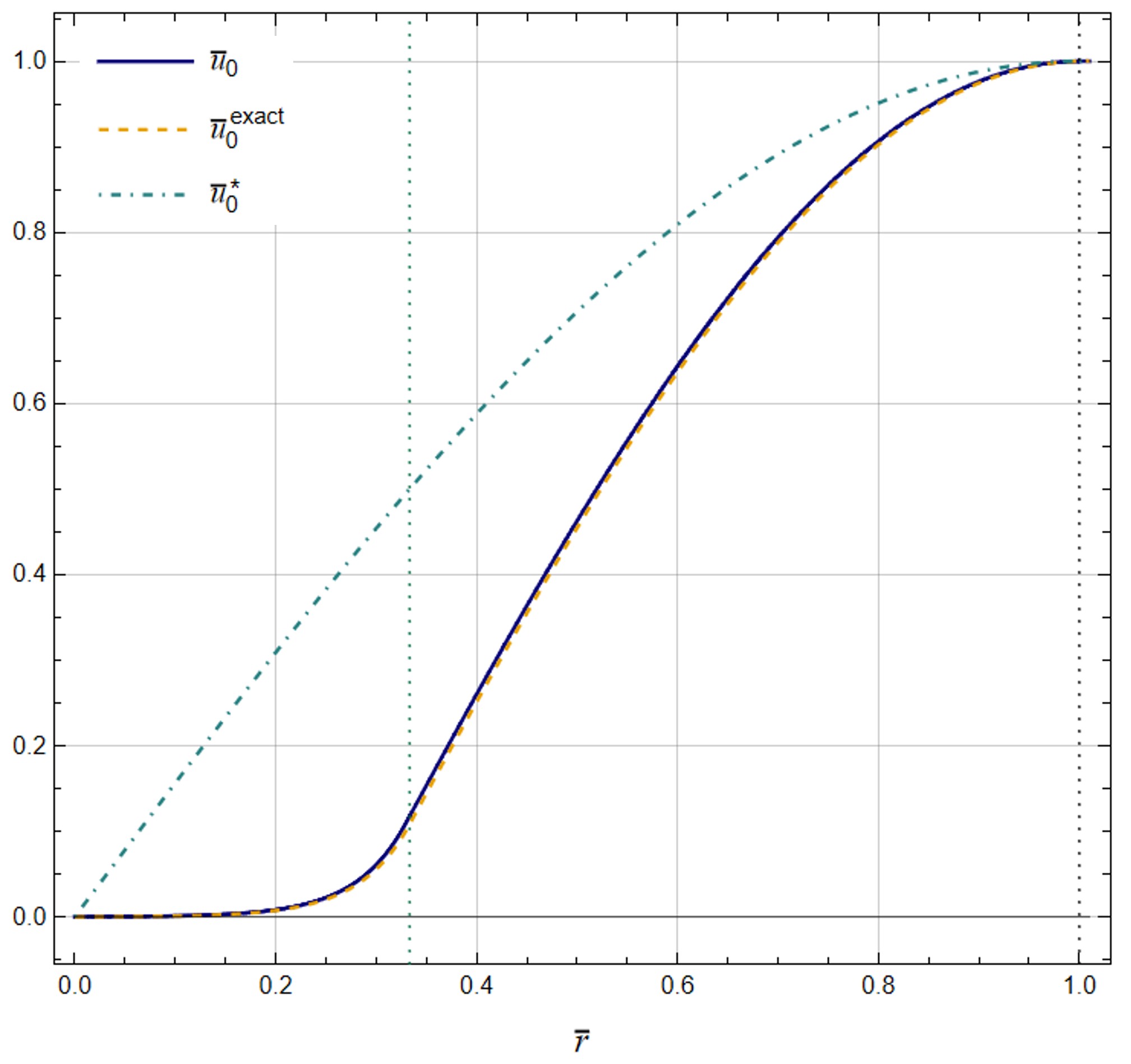}
    \caption{Comparison between the analytical solution in Eq.~\eqref{eq:u0_asymp} (solid line), the numerical solution (dashed line, obtained with \texttt{Mathematica} following \cite{Macedo-Lima:2023fzp}), and the solution $\bar{u}_0^*(\bar{r})=\sin(\pi/2\,\bar{r})$. The vertical dotted lines indicate the end of the repulsive region ($\bar{r}=\bar{r}_1$) and the range of the potential ($\bar{r}=1$), respectively.}
    \label{fig:u0_as}
\end{figure}

This potential is also of interest because it is analytically solvable in the limits $R/a\to 0$ (unitarity limit), $U_1\gg U_2$, and $\sqrt{U_1}r_1\gg1$. This regime provides a useful approximation for more complex interactions with a strong repulsive core. The solution of Eq. \eqref{eq:well} that respects $\bar{u}_0(0)=0$ is
\begin{equation}
    \bar{u}_0(\bar{r})=\begin{cases}
        A \sinh(k_1\,\bar{r}) & \bar{r}\in[0,\bar{r}_1)\\[4pt]
        B\sin(k_2\,\bar{r}+\delta) & \bar{r}\in [\bar{r}_1,1)\\[4pt]
        1 & \bar{r}\geq 1,
    \end{cases}
\end{equation}
where $k_1=\sqrt{\bar{U}_1}$ and $k_2=\sqrt{\bar{U}_2}$. Continuity of the first derivative at $\bar{r}=1$ gives
\begin{equation}
    Bk_2\cos(k_2+\delta)=0\,
\end{equation}
so
\begin{equation}
    \delta=\frac{\pi}{2}-k_2 \Rightarrow \bar{u}_0(\bar{r})=B\cos\big(k_2(1-\bar{r})\big)\quad \bar{r}\in[\bar{r}_1,1)\,.
\end{equation}
Continuity of the function at $\bar{r}=1$ then yields $B=1$.

At $\bar{r}=\bar{r}_1$, continuity of the function implies
\begin{equation}
    A\sinh(k_1\bar{r}_1)=\cos\big(k_2 (1-\bar{r}_1)\big)
    \,\Rightarrow \,
    A=\frac{\cos\big(k_2 (1-\bar{r}_1)\big)}{\sinh(k_1\bar{r}_1)}.
\label{eq:A}
\end{equation}
Continuity of the logarithmic derivative gives
\begin{equation}
    k_1 \coth(k_1 \bar{r}_1)=k_2 \tan\big(k_2 (1-\bar{r}_1)\big)\,.
\label{eq:log_cont}
\end{equation}
This equation provides the relation between $V_1$ and $V_2$ required to achieve an infinite scattering length. For instance, we can express $V_2$ as a function of $V_1$ by solving Eq. \eqref{eq:log_cont} under the aforementioned assumptions, yielding
\begin{equation}
    k_2=\frac{\pi}{2(1-\bar{r}_1)}
    \left(1-\frac{1}{(1-\bar{r}_1)k_1}\right).
\end{equation}
The leading term coincides with that of a square well of width $1-\bar{r}_1$. Using
\begin{equation}
    \cos\arctan(x)=\frac{1}{\sqrt{1+x^2}},
\end{equation}
together with Eq. \eqref{eq:log_cont}, we obtain
\begin{equation}
    A
    =\frac{1}{\sinh(k_1 \bar{r}_1)}
    \frac{\pi}{2(1-\bar{r}_1)k_1}.
\end{equation}
The complete zero-energy wave function is therefore
\begin{equation}
    \bar{u}_0(\bar{r})=\begin{cases}
        \dfrac{1}{\sinh(k_1 \bar{r}_1)}
        \dfrac{\pi}{2(1-\bar{r}_1)k_1}
        \sinh(k_1\bar{r}) & \bar{r}\in[0,\bar{r}_1)\\[12pt]
        \cos\big(k_2 (1-\bar{r})\big) & \bar{r}\in[\bar{r}_1,1)\\[4pt]
        1 & \bar{r}\geq 1.
    \end{cases}
    \label{eq:u0_asymp}
\end{equation}
We may now compute
\begin{equation}
    \bar{r}_0=2\int_0^1 \big(1-\bar{u}_0^2(\bar{r})\big)\,d\bar{r}\simeq 1+\bar{r}_1\,.
\end{equation}
In Fig. \ref{fig:u0_as}, we compare the wave function \eqref{eq:u0_asymp} for $\bar{U}_1=400$ and $\bar{r}_1=1/3$ with that of a purely attractive spherical well of width $R$ in the unitarity limit ($\bar{U}^*_w=\pi^2/4$), which has $\bar{u}_0^*(\bar{r})=\sin(\pi/2\,\bar{r})$ and $\bar{r}_0=1$. In this case, it is evident how the wave function is exponentially suppressed in the classically forbidden region, where the $\sinh$ behavior replaces the $\sin$ one. These findings are consistent with the limit of an infinite barrier, which forces $\bar{u}_0=0$ for $\bar{r}<\bar{r}_1$, resulting in $\bar{r}_0=1+\bar{r}_1$.

\subsubsection{Yukawa potential}
\begin{figure}[h]
    \centering
    \includegraphics[width=0.9\linewidth]{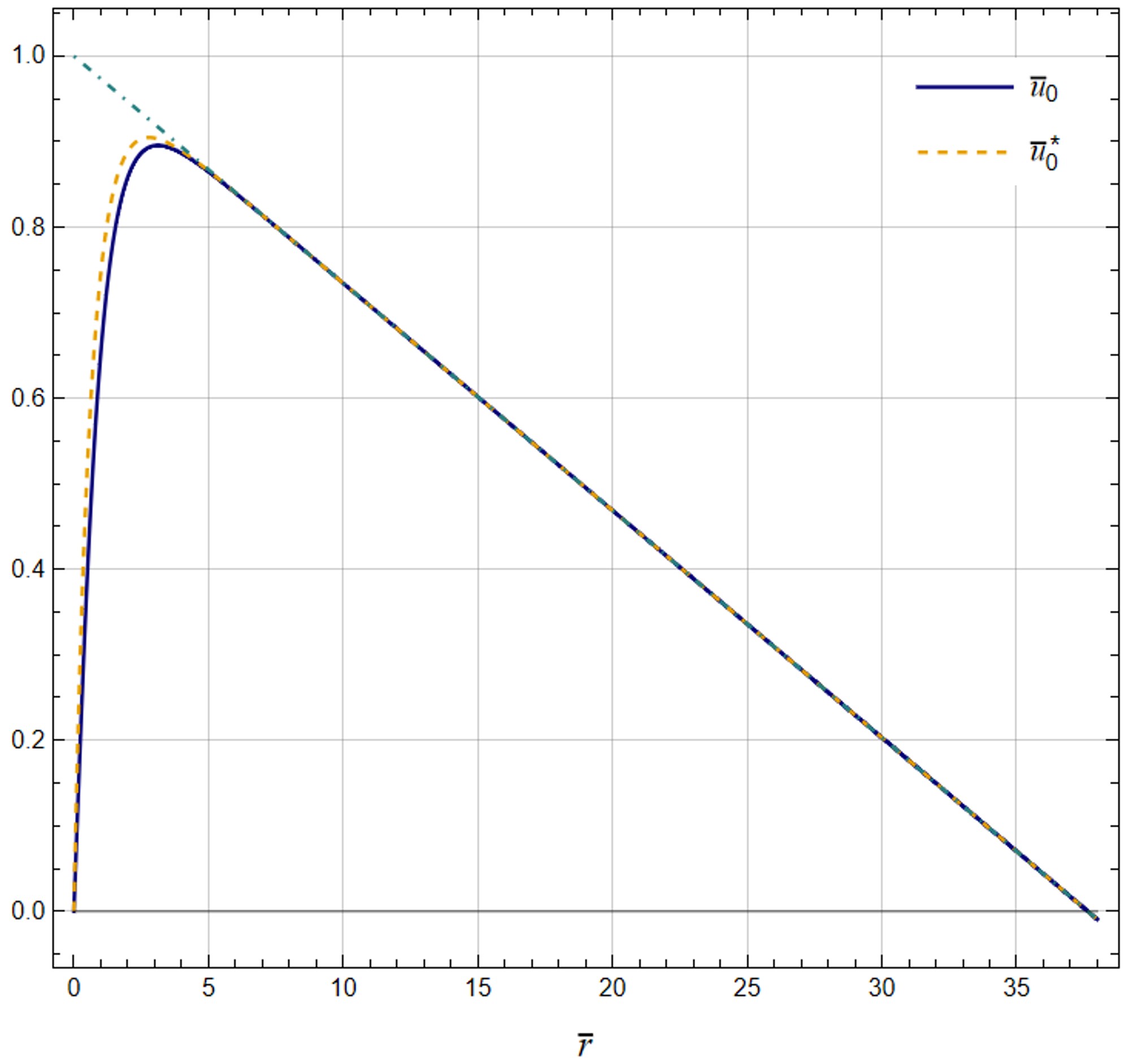}
    \caption{Solid line: zero-energy wavefunction for the potential in Eq.~\eqref{eq:yukawa_pot} with the parameters listed in Tab.~\ref{tab:yukawa}. The dashed line indicates the solution for the purely attractive problem. The dash-dotted line represents the asymptotic solution $1-\bar{r}/\bar{a}$, which is common to both solutions since they share the same scattering length.}
    \label{fig:u0_yuk}
\end{figure}
\begin{figure}[h]
    \centering
    \includegraphics[width=0.9\linewidth]{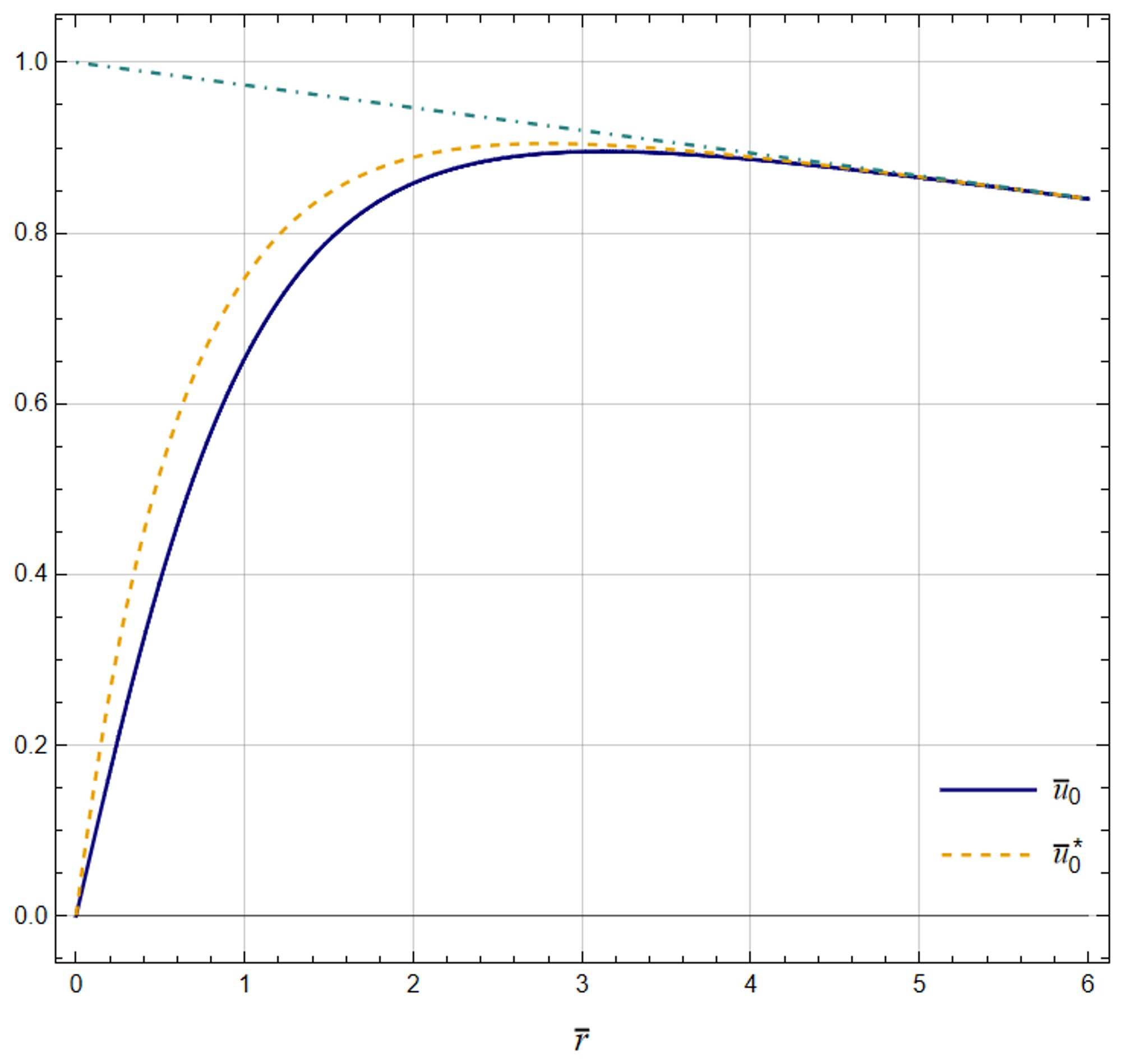}
    \caption{Zoom in of the two wave functions to highlight how $\bar{u}_0$ lies below $\bar{u}_0^*$. Furthermore, it illustrates the rapid convergence to the asymptotic solution due to the steep exponential decay of the potentials.}
    \label{fig:u0_yuk_zoom}
\end{figure}

Although the Yukawa potential does not strictly possess a finite range, its exponential decay allows us to truncate it at an arbitrarily large $R$ and verify that the results are largely independent of this cutoff.\footnote{Restricting the potential to a compact support also prevents non-analyticity issues in the complex $k$ plane \cite{osti_4661960}.} We therefore consider the potential
\begin{equation}
    V_y(r)=\frac{\hbar^2}{2mr_2}\left(g_1^2\,\frac{e^{-r/r_1}}{r}-g_2^2\,\frac{e^{-r/r_2}}{r}\right)\theta(R-r)\,,
    \label{eq:yukawa_pot}
\end{equation}
where $\theta(x)$ is the Heaviside step function and $r_1<r_2$, ensuring the potential has an attractive tail. We also choose $R \gg r_2$ to be well within the exponentially decaying region. As $r\to0^+$, the potential behaves as
\begin{equation}
    V(r)\sim \frac{g_1^2-g_2^2}{r} \,,
\end{equation}
so that if $|g_1|>|g_2|$, it features a repulsive core.
\begin{table}[ht]
    \centering
    \begin{tabular}{|ccccc|ccc|}
    \hline
         $g_1^2$& $g_2^2$ & $\bar{r}_1$ & $\bar{R}$ & $g^{*2}$ & $\bar{a}$ & $\bar{r}_0$ & $\bar{r}_0^*$\\
         \hline\hline
         3 & 2.7 & 0.4 & 30 & 1.7416 & 37.6 & 1.23 & 1.01 \\
         \hline
    \end{tabular}
    \caption{Parameters used for the Yukawa problem and the corresponding scattering length and effective range.}
    \label{tab:yukawa}
\end{table}
In this case, it is more convenient to use $r_2$ as the natural length scale, defining the dimensionless variable
\begin{equation}
    \bar{r}=\frac{r}{r_2} \,,
\end{equation}
and similarly for all other dimensional quantities. Using the parameters in Tab. \ref{tab:yukawa}, we obtain $\bar{a}=37.6$ and $\bar{r}_0=1.23$. Notably, varying the cutoff $\bar{R}$ by a factor of two does not significantly alter these results.\footnote{For instance, setting $\bar{R}=6$ changes $\bar{a}$ by only $2\%$. This indicates that, due to the exponential decay, the wave function quickly converges to the free solution.} As before, we can compare this outcome with a purely attractive Yukawa potential $V_y^*(r)\propto -g^{*\,2}\,e^{-r/r_2}/r$ yielding the same scattering length. Setting $g^*$ as shown in Tab. \ref{tab:yukawa}, we find $\bar{a}^*=\bar{a}$ but $\bar{r}_0^*=1.01$, corroborating the general trend of an increased effective range observed in the previous section.

\section{Conclusions}
\label{sec:conc}
In this work, we investigated the fundamental constraints on the sign of the effective range $r_0$ in low-energy scattering. While it is a standard result that purely attractive potentials supporting a shallow bound state yield a positive effective range \cite{Landau:1991wop,Smorod}, we extended this analysis to a more general and physically motivated class of interactions. Specifically, we proved that for local, finite-range potentials characterized by an inner repulsive core and an attractive tail, the effective range remains strictly positive as long as the scattering length is larger than the potential range ($a > R$).

We illustrated this theorem using two representative models: a piecewise continuous spherical barrier followed by a well potential, and a truncated superposition of attractive and repulsive Yukawa potentials. In both instances, the presence of the repulsive core suppresses the physical wave function in the inner, classically forbidden region. This suppression consistently increases the value of the effective range compared to purely attractive equivalent potentials, confirming that local repulsion cannot drive $r_0$ to negative values.

These findings carry significant implications for modern hadron physics, particularly in the phenomenological study of exotic states. Our results demonstrate that introducing an arbitrary repulsion into a single-channel local potential is fundamentally insufficient to generate the negative effective range typically associated with compact multiquark configurations without invoking explicit coupled-channel dynamics or energy-dependent interactions \cite{Baru:2021ldu,Esposito:2023mxw,Esposito:2025hlp,Esposito:2021vhu}.

\bibliographystyle{unsrt}
\bibliography{biblio}

\end{document}